\documentclass[a4paper]{amsart}
\usepackage[utf8]{inputenc}
\usepackage[T1]{fontenc}
\usepackage{amssymb,todonotes,xfrac,mathrsfs,pdftexcmds,dutchcal}
\usepackage{frcursive}
\usepackage{hyperref}
\author{Przemysław Koprowski}
\title[Factorization and root-finding\dots]{Factorization and root-finding for polynomials over division quaternion algebras}
\newtheorem{alg}{Algorithm}
\newtheorem{thm}{Theorem}
\newtheorem{cor}[thm]{Corollary}
\newtheorem{lem}[thm]{Lemma}

\newtheorem{prop}[thm]{Proposition}
\theoremstyle{remark}
\newtheorem{rem}{Remark}
\newtheorem{exm}{Example}
\newenvironment{poc}[1][]{\begin{proof}[\ifthenelse{\equal{#1}{}}{Proof of correctness}{Proof of correctness of #1}]}{\end{proof}}
\newcommand{\mycal}[1]{\mathfrak{#1}}%{\mbox{\textcursive{#1}}}
\newcommand{\pp}{p}
\newcommand{\pq}{q}
\newcommand{\pr}{r}
\newcommand{\qpp}{\mycal{p}}
\newcommand{\qpq}{\mycal{q}}
\newcommand{\qpr}{\mycal{r}}
\newcommand{\qa}{\mycal{a}}
\newcommand{\qb}{\mycal{b}}
\newcommand{\qc}{\mycal{c}}
\newcommand{\qi}{\mycal{i}}
\newcommand{\qj}{\mycal{j}}
\newcommand{\qk}{\mycal{k}}
\newcommand{\CA}{\mathfrak{A}}
\newcommand{\CL}{\mathscr{L}}
\newcommand{\CR}{\mathscr{R}}
\newcommand{\HH}{\mathbb{H}}
\newcommand{\QQ}{\mathbb{Q}}
\newcommand{\RR}{\mathbb{R}}
\newcommand{\ZZ}{\mathbb{Z}}
\newcommand{\K}{K}
\renewcommand{\L}{L}
\newcommand{\Ax}{\CA[x]}
\newcommand{\charpoly}[1]{\chi_{#1}}
\newcommand{\conj}[1]{\overline{#1}}
\newcommand{\ext}[2]{#1/#2}
\newcommand{\form}[1]{\langle #1\rangle}
\newcommand{\gena}{\alpha}
\newcommand{\genb}{\beta}
\newcommand{\ideal}[1]{(#1)}
\newcommand{\Kx}{\K[x]}
\DeclareMathOperator{\lc}{lc}
\DeclareMathOperator{\lclm}{lclm}
\newcommand{\norm}[2][]{N_{#1}#2}
\newcommand{\Pform}[1]{\langle\langle #1\rangle\rangle}
\DeclareMathOperator{\gcrd}{gcrd}
\newcommand{\quo}[2]{\sfrac{#1}{#2}}
\newcommand{\quat}[3]{\bigl(\frac{#1,#2}{#3}\bigr)}
\newcommand{\Quat}{\quat{\gena}{\genb}{\K}}
\newcommand{\term}[1]{\emph{#1}}
\DeclareMathOperator{\Tr}{Tr}
\newcommand{\un}[1][K]{#1^{\times}}
\newcommand{\Magma}{\textsc{Magma}}
\newcommand{\qPoly}{\mbox{\smash{\scalebox{1.0}[0.8]{\textcursive{qPoly}}}}}
\begin{document}
%\linenumbers % turn on line numbering
\begin{abstract}
Polynomial factorization and root finding are among the most standard themes of computational mathematics. Yet still, little has been done for polynomials over quaternion algebras, with the single exception of Hamiltonian quaternions for which there are known numerical methods for polynomial root approximation. The sole purpose of the present paper is to present a polynomial factorization algorithm for division quaternion algebras over number fields, together with its adaptation for root finding.
\end{abstract}
\maketitle

\section{Introduction}
Polynomial factorization is one of the classical subjects in the realm of computational algebra. Over the decades, numerous algorithms have been developed for factoring polynomials over various base rings, like integers and rationals, finite fields, local and global fields, etc. Nonetheless, there are still many important classes of rings for which there are no known methods for polynomial factorization. One of these notable omissions is the class of quaternion algebras. Although quaternion polynomials are a classical subject that has been studied since the first half of the 20th century, the algorithmic side of the theory is strongly underdeveloped. To the best of our knowledge, the only algorithms developed so far concentrate exclusively on polynomials over real quaternion algebras with an emphasis on the Hamiltonian quaternions (see, e.g., \cite{JO2010, Kalantari2013, Niven1941, SKS2019, SP2001} for root-finding algorithms, or \cite{LSSS2023, LSS2019, SS2021} for factorization algorithms). Unfortunately, the methods developed for real quaternions cannot be applied in a general setting. The fundamental reason is that the Hamiltonian quaternions admit the fundamental theorem of algebra. It is not a case for general division quaternion algebras. The author is unaware of even a single paper addressing the factorization of polynomials over \emph{general} quaternion algebras. Thus, the sole purpose of this paper is to partially remedy this situation and to present algorithms for factoring polynomials over division quaternion algebras over number fields (see Algorithm~\ref{alg:complete_factor}) and to compute roots of polynomials in such algebras (see Algorithm~\ref{alg:root_finding}).

There are a few different ways how one can define polynomials over a general associative ring. If the ring is commutative, they all coincide, but no longer so when the ring of coefficients is non-commutative. In this paper, we follow the most classical approach, and by a polynomial, over an associative ring~$\CA$, we understand a sum $\qpp = \qa_0 + \qa_1 x + \dotsb + \qa_n x^n$, where the coefficients $\qa_0, \dotsc, \qa_n$ are taken from $\CA$ and the indeterminate~$x$ commutes with the coefficients. These are called \term{left polynomials} in \cite{GM1965, GM1966} or \term{unilateral polynomials} in more recent sources like \cite{FMSS2018, Kalantari2013, SP2001}. 

Throughout this paper, we use the following notational convention. Commutative rings and their elements are typeset using the standard typeface, while for non-commutative rings and their elements, we use fraktur letters. Further, $\K$ will always denote a number field, i.e. a finite extension of the field~$\QQ$ of the rationals. Recall that an algebra~$\CA$ over~$\K$ is called a \term{quaternion algebra} if it has a basis of the form $\{1, \qi, \qj, \qk\}$ subject to the relations 
\[
\qi^2 = \gena,\qquad 
\qj^2 = \genb,\qquad 
\qi\qj = \qk = -\qj\qi,
\]
for some nonzero elements $\gena, \genb\in \K$. The algebra~$\CA$ is then denoted $\Quat$. For an in-depth presentation of the theory of quaternion algebras, we refer the reader to \cite{Vigneras1980} and \cite{Voight2021}. We will always identify the base field~$\K$ with the subalgebra $\K\cdot 1$ of~$\CA$, which is known to form the center $Z(\CA)$ of~$\CA$. We shall write $\un$ (respectively $\CA^\times$) to denote the multiplicative group of nonzero elements of~$\K$ (resp. invertible elements in~$\CA$). 

Every quaternion algebra is equipped with a standard involution that will be denoted by an overbar in what follows. For a quaternion $\qa\in \CA$, the product $\qa\conj{\qa}$ is called the \term{norm} (in \cite{Lam2005}) or the \term{reduced norm} (in \cite{Vigneras1980, Voight2021}) of~$\qa$ and will be denoted $\norm{\qa}$. Likewise, the sum $\qa + \conj\qa$ is called the \term{\textup(reduced\textup) trace} of~$\qa$, denoted $\Tr\qa$. The natural inclusion $\quat{\gena}{\genb}{\K}[x]\hookrightarrow \quat{\gena}{\genb}{\K(x)}$ let us extend the notion of the norm and standard involution to quaternionic polynomials. In particular, if $\qpp = \qa_0 + \qa_1x + \dotsb + \qa_nx^n$, then $\conj{\qpp} = \conj{\qa_0} + \conj{\qa_1}x + \dotsb + \conj{\qa_n}x^n$ and $\norm{\qpp} = \qpp\conj{\qpp}$. We will use the same symbol $\norm{}$ to denote also the standard norm in a field extension and its natural prolongation to the polynomial rings. In case of ambiguity, we will indicate the corresponding extension in the norm's subscript.

A quaternion algebra is either a division algebra or is isomorphic to the $2\times 2$ matrix algebra $M_2\K$ (see, e.g., \cite[Theorem~5.4.4]{Voight2021}). The latter case happens if and only if the $2$-fold Pfister form $\Pform{-\gena, -\genb} = x_0^2 - \gena x_1^2 - \genb x_2^2 + \gena\genb x_3^2$ is \term{isotropic} (i.e. it has a non-trivial zero). The quaternion algebra $\Quat$ is then said to be \term{split}. 

\section{Preliminaries}
Observe that over a split quaternion algebra, every polynomial $\pp\in \Kx$ has a trivial factorization. Indeed, if we identify $\CA$ with $M_2\K$ and $\K$ with its subfield of scalar matrices, we can write 
\[
\begin{pmatrix} \pp & 0\\ 0 & \pp \end{pmatrix}
= \begin{pmatrix} \pp & 0\\ 0 & 1 \end{pmatrix}
  \cdot \begin{pmatrix} 1 & 0\\ 0 & \pp \end{pmatrix}
= \begin{pmatrix} \pp & 0\\ 0 & 1 \end{pmatrix}
  \cdot \conj{\begin{pmatrix} \pp & 0\\ 0 & 1 \end{pmatrix}}.
\]
Notice that both divisors of~$\pp$ constructed above have the same degree as the polynomial~$\pp$ itself. This is due to the fact that all their coefficients, save the constant terms, are zero-divisors. Consequently, in the rest of the paper, we concentrate entirely on polynomials over division quaternion algebras. 

Below we collect some known facts concerning polynomials over division quaternion algebras. They are all very classical but scattered throughout the literature and sometimes expressed using different terminology. For a general introduction to polynomial rings over division rings (including division quaternion algebras), we refer the reader to \cite[\S16]{Lam2001}.

\begin{thm}[Ore]
If $\CA$ is a division quaternion algebra, then the ring $\Ax$ of quaternionic polynomials admits a right-hand division with a remainder and a right-hand Euclidean algorithm. 
\end{thm}

In particular, for any two quaternionic polynomials $\qpp, \qpq\in \Ax$, there exists a unique monic greatest common right divisor (denoted $\gcrd(\qpp, \qpq)$ hereafter) and the least common left multiple (denoted $\lclm(\qpp, \qpq)$). In fact, \cite[Theorem~I.8]{Ore1933} provides an explicit formula for constructing the least common left multiple from the sequence of the right reminders computed by the right-hand Euclidean algorithm. However, we will not use this formula in the present paper.

\begin{thm}[Ore] 
If $\CA$ is a division quaternion algebra, then every monic polynomial $\qpp\in \Ax$ can be expressed as a product of irreducible quaternionic polynomials. Every two such factorizations of~$\qpp$ have the same number of factors.
\end{thm}

For proofs of the above two theorem see \cite[Chapter~I, \S2 and Theorem~II.1]{Ore1933}. Of course, in general, a quaternionic polynomial may have infinitely many different factorizations, even assuming that all the factors are monic. To observe this phenomenon take any non-constant monic polynomial $\pp\in \Kx$ with coefficients in the ground field. Assume that $\pp$ factors into a product $\pp = \qpp\qpq$ of two monic quaternionic polynomials $\qpp \qpq\in \Ax$, not necessarily irreducible. Then for every invertible quaternion $\qa \in \CA$ we have 
\[
\pp = \qa\pp\qa^{-1} = (\qa\qpp \qa^{-1})(\qa\qpq \qa^{-1}).
\]
If only $\qpp, \qpq\notin \Kx$, this way we obtain another factorization of~$\pp$, where $\qa\qpp\qa^{-1}$ and $\qa\qpq\qa^{-1}$ are again monic. 

\begin{prop}[Beck]\label{prop:Beck}
Let $\CA$ be a division quaternion algebra over~$\K$. Then every polynomial $\qpp\in \Ax$ can be uniquely expressed as a product $\qpp = \qc\cdot \qpq\cdot \pp$, where $\qc \in \CA^\times$ is the leading coefficient of~$\qpp$, $\pp$ is a monic polynomial with coefficients in~$\K$ and $\qpq\in \Ax$ is a monic quaternionic polynomial not divisible by any non-constant polynomial from $\Kx$. Moreover, if $\Ax$ is treated as a free module of rank~$4$ over $\Kx$ with a basis $\{ 1, \qi, \qj, \qk\}$, then $\pp$ is the greatest common divisor \textup(in~$\Kx$\textup) of the coordinates of~$\qpp$.
\end{prop}

The polynomial~$\pp$ described in the last proposition is called the \term {maximal central factor} of~$\qpp$. A proof of Proposition~\ref{prop:Beck} can be found in \cite[R\'esultat~4]{Beck1979}. In view of this proposition, it is clear that in order to factor quaternionic polynomials, it suffices to have a procedure that factors in $\Ax$ polynomials that are irreducible in $\Kx$, and another procedure that factors quaternionic polynomials not divisible by polynomials from $\Kx$. 

\begin{lem}
Let $\CA = \Quat$ be a division quaternion algebra over~$\K$ and $\qa, \qb\in \CA$ be two quaternions. If $\qa\qb\in Z(\CA) = \K$, then the following two conditions hold: 
\begin{enumerate}
\item $\qa\qb = \qb\qa$;
\item  there are scalars $c, d \in \K$ such that $c\cdot \qb = d\cdot \conj{\qa}$
\end{enumerate}
\end{lem}

\begin{proof}
If $\qa = 0$ or $\qb = 0$, then the assertions hold trivially. Hence, without loss of generality, we can assume that $\qa, \qb\in \CA^\times$ and so their product $d := \qa\qb$ is nonzero. Compute 
\[
\qb\qa 
= \qa^{-1}\qa\qb\qa 
= \qa^{-1}\cdot d\cdot \qa 
= d\cdot \qa^{-1}\qa = \qa\qb.
\]
This proves the first assertion. The second one follows immediately. Indeed, from $\qa\qb = d$ we infer
\[
\norm(\qa)\cdot \qb 
= \norm(\qa)\cdot \qa^{-1}\cdot d = d\cdot \conj{\qa}.
\qedhere
\]
\end{proof}

\begin{rem}
The assumption that $\CA$ is a division algebra is indispensable. To observe that, take the split quaternion algebra $\CA = M_2\K$ and let 
\[
\qa := \begin{pmatrix} 0 & 0\\ 0 & 1 \end{pmatrix}
\qquad\text{and}\qquad
\qb := \begin{pmatrix} 1 & 1\\ 0 & 0 \end{pmatrix}.
\]
Then $\qa\qb$ is the zero matrix, hence in particular $\qa\qb\in Z(\CA)$ but
\[
\qb\qa = \begin{pmatrix} 0 & 1\\ 0 & 0 \end{pmatrix}\neq \qa\qb.
\]
\end{rem}

\begin{prop}\label{prop:factorization_of_irreducible}
Let $\CA = \Quat$ be a quaternion division  algebra over~$\K$. Further, let $\pp\in \Kx$ be a monic irreducible polynomial with coefficients in~$\K$. Then either $\pp$ remains irreducible in $\Ax$ or it factors into a product $\pp = \qpp\conj{\qpp}$ for some irreducible quaternionic polynomial $\qpp\in \Ax$.
\end{prop}

\begin{proof} 
Suppose that $\pp$ is not irreducible in $\Ax$. Let $\pp = \qpp\qpq$ be some factorization of~$\pp$ with $\qpp,\qpq\in \Ax$. Applying the previous lemma to the quaternion algebra $\quat{\gena}{\genb}{\K(x)}$ and using the fact that $\qpp$, $\qpq$ and $\pp$ are all polynomials we obtain that 
\[
\pp = \qpp\conj{\qpp}\cdot \pq
\]
for some polynomial $\pq\in \Kx$. Since, the norm $\norm{\qpp} = \qpp\conj{\qpp}$ also sits in $\Kx$, it follows from the irreducibility of~$\pp$ that $\deg \pq = 0$, i.e. $\pq\in \K$. 

It remains to prove that $\qpp$ is irreducible in $\Ax$. Suppose that there are some non-constant polynomials $\qpq, \qpr\in \Ax$, such that $\qpp = \qpq\qpr$. Therefore
\[
\pp 
= \qpq\qpr\cdot \conj{\qpq\qpr} 
= \qpq\cdot \qpr\conj{\qpr}\cdot \conj{\qpq}
= \qpq\conj{\qpq}\cdot \qpr\conj{\qpr}.
\]
But this contradicts the irreducibility of~$\pp$ as $\qpq\conj{\qpq}, \qpr\conj{\qpr}\in \Kx$.
\end{proof}

We may now present a criterion of irreducibility in $\Ax$.

\begin{prop}\label{prop:ireeducibility_criterium}
Let $\CA = \Quat$ be a division quaternion algebra over a number field~$\K$. Further, let $\pp\in \Kx$ be a monic irreducible polynomial and $\L := \quo{\Kx}{\ideal{\pp}}$. Then $\pp$ remains irreducible in $\Ax$ if and only if $\CA\otimes \L$ does not split.
\end{prop}

\begin{proof}
Proposition~\ref{prop:factorization_of_irreducible} asserts that $\pp$ is irreducible in $\Ax$ if and only if it does not equal the norm $\norm{\qpp} = \qpp\conj{\qpp}$ of any polynomial $\qpp\in \Ax$. The latter condition is equivalent to the one that $\pp$ is not represented over $\Kx$ by the Pfister form $\Pform{-\gena, -\genb} = \form{1, -\gena, -\genb, \gena\genb}$, which is the norm form of~$\CA$. Now, since $\pp$ is monic by assumption, it follows from \cite[Proposition~3]{Pourchet1971} (or \cite[Lemma~17.3]{Rajwade1993}) that $\pp$ is not represented by this form if and only if $L$ is not the splitting field for~$\CA$. 
\end{proof}

\begin{rem}
A procedure for testing if a quaternion algebra over a number field splits is described in \cite{Voight2013} and readily available in some computer algebra systems like Magna~\cite{BCP1997} and Sage~\cite{sagemath}. Hence, the above irreducibility test can be easily implemented. An example implementation for the computer algebra system \Magma{}~\cite{BCP1997} was prepared by the author (see the closing section).
\end{rem}

Comparing the degrees, we obtain the following immediate consequence of Proposition~\ref{prop:ireeducibility_criterium} that can be used as a quick-exit in the procedure for testing the irreducibility. 

\begin{cor} 
Let $\CA$ be a division quaternion algebra and $\pp\in \Kx$ be a monic, irreducible polynomial of odd degree, then $\pp$ remains irreducible in $\Ax$.
\end{cor}

To conclude this section, we present a complete irreducibility criterion for quaternionic polynomials that can be directly implemented in a computer algebra system.

\begin{thm}
Let $\CA$ be a division quaternion algebra over a number field~$\K$. Further, Let $\qpp\in \Ax$ be a non-constant polynomial and $\qpp = \qc\cdot \qpq\cdot \pp$ be its decomposition as in Proposition~\ref{prop:Beck}. 
\begin{enumerate}
\item If $\qpq = 1$, then $\qpp$ is irreducible in~$\Ax$ if and only if $\pp$ is irreducible in $\Kx$ and $\CA\otimes \L$ does not split, where $\L := \quo{\Kx}{\ideal{\pp}}$.
\item\label{it:irred_q} If $\pp = 1$, then $\qpp$ is irreducible in~$\Ax$ if and only if the norm $\norm\qpq$ of~$\qpq$ is irreducible in~$\Kx$. 
\item If $\qpq\neq 1$ and $\pp\neq 1$, then $\qpp$ is reducible in~$\Ax$.
\end{enumerate}
\end{thm}

\begin{proof}
The first assertion follows immediately from Proposition~\ref{prop:ireeducibility_criterium}, while the third one is trivial. We must prove only assertion~\eqref{it:irred_q}. Without loss of generality, we can assume that $\qpp$ is monic. Hence, $\qpp = \qpq$. If $\qpp$ is reducible, say $\qpp = \qpp_1\qpp_2$ then $\norm\qpq = \norm{\qpp_1}\cdot \norm{\qpp_2}$ is reducible in~$\Kx$, too. This proves one implication. Now, suppose that $\qpp$ is irreducible but $\norm\qpq$ has a factorization $\norm\qpq = \pq_1\pq_2$ for some non-constant polynomials $\pq_1, \pq_2\in \Kx$. The maximal central factor of~$\qpp$ is trivial by assumption. Therefore, neither $\qpp$ nor $\conj{\qpp}$ is divisible by any of the polynomials $\pq_1, \pq_2$. But we have $\qpp\conj{\qpp} = \norm\qpq = \pq_1\pq_2$. Therefore, \cite[p. 494]{Ore1933} says that $\qpp$ is right-divisible by $\lclm(\pq_2, \conj{\qpp})\cdot \conj{\qpp}^-1 = \pq_2$. This contradiction implies that $\norm\qpq$ must be irreducible, proving assertion~\eqref{it:irred_q}.
\end{proof}

\section{Factorization of central polynomials}\label{sec:central}
Our first task is to factor in $\Ax$ an irreducible polynomial~$p$ with coefficients in the base field~$\K$. Before we consider this problem in its full generality, we first deal with a special case. If there is a maximal (commutative) subfield $\K_*$ of~$\CA$ such that $\K_*$ embeds into $\L := \quo{\Kx}{\ideal\pp}$, then a factorization of~$\pp$ in $\Ax$ is exceptionally easy.

\begin{lem}\label{lem:K*_into_L}
Let $\CA$ be a division quaternion algebra over a number field~$\K$ and $\pp\in \Ax$ be a monic irreducible polynomial of degree $\deg\pp > 1$. Let $\L := \quo{\Kx}{\ideal\pp}$. If there is a maximal commutative subfield $\K_*$ of~$\CA$ such that $\pp = \norm[\ext{\K_*}{\K}]{(\pq)}$ for some polynomial $\pq\in \K_*[x]$, then
\begin{enumerate}
\item $\K_*$ embeds into~$\L$; 
\item $\pp$ has a non-trivial factorization in~$\Ax$;
\item the quaternion algebra $\CA\otimes\L$ is split.
\end{enumerate}
\end{lem}

\begin{proof}
Let $\K_* = \K(\sqrt{d})$ be such that $\pp = \norm[\ext{\K_*}{\K}]{(\pq)}$ for some $\pq \in \K_*[x]$. Write $\pq = \pq_0 + \pq_1\sqrt{d}$. Then $\pp = \pq_0^2 - \pq_1^2 d$. Let $a_i := \pq_i + \ideal\pp\in \L$, for $i \in \{0,1\}$. We have 
\[
0 = a_0^2 - a_1^2 d
\]
and so $d$ is a square in~$\L$. This means that $\K_*$ embeds into~$\L$. Now, $\pp = (\pq_0 + \pq_1\sqrt{d})(\pq_0 - \pq_1\sqrt{d})$, where $\pq_0, \pq_1\in \K_*[x]\subset \Ax$. Therefore, $p$ has a non-trivial factorization in $\Ax$. The third assertion follows now from Proposition~\ref{prop:ireeducibility_criterium}.
\end{proof}
 
\begin{alg}\label{alg:subfield}
Let $\CA$ be a division quaternion algebra over a number field~$\K$, and let $\pp\in \Kx$ be a monic irreducible polynomial. If there is a maximal commutative subfield of~$\CA$  over which $\pp$ factors, then this algorithm outputs a factorization of~$\pp$ over~$\CA$. Otherwise, it reports a failure.
\begin{enumerate}
\item Build a field extension $\L := \quo{\Kx}{\ideal\pp}$.
\item\label{st:subfield:subfields} Construct all the subfields $\L_1, \dotsc, \L_n$ of~$\L$ of degree $2$ over~$\K$ \textup(see Remark~\ref{rem:subfields} below\textup).
\item Repeat the following steps for every $\L_i$:
  \begin{enumerate}
   \item Check if the quaternion algebra $\CA\otimes \L_i$ splits. If it does not, then reiterate the loop. 
    \item Find an embedding $\varphi: \L_i\hookrightarrow \CA$ \textup(see Remark~\ref{rem:embedding} below\textup) and let $\K_* := \varphi(\L_i)$.
    \item Factor $\pp$ in $\K_*[x]$ into the product $\pp = \qpq\cdot \conj{\qpq}$, where $\qpq\in \K_*[x] \subset \Ax$.
    \item\label{st:subfield:success} Output $( \qpq, \conj{\qpq})$ and quit.
  \end{enumerate}
\item Report a failure.
\end{enumerate}
\end{alg}

\begin{poc}
Let $\K_*$ is a maximal (commutative) subfield of~$\K$ such that $\pp$ has a non-trivial factorization in $\K_*[x]$, then $(\K_*: \K) = 2$  and so $\K_*$ embeds into~$\L$ by Lemma~\ref{lem:K*_into_L}. Hence, $\K_*$ is isomorphic to one of the fields $\L_1, \dotsc, \L_n$ constructed in step~\eqref{st:subfield:subfields} of the algorithm. Therefore, a proper factorization of~$\pp$ is returned in step~\eqref{st:subfield:success}. Conversely, assume that $\L_i$ is a subfield of~$\L$, quadratic over~$\K$, and such that $\CA\otimes \L_i$ is split. Then, \cite[Theoreme~I.2.8]{Vigneras1980} asserts that $\L_i$ embeds into~$\CA$. Hence, up to an isomorphism, it is a maximal subfield of~$\CA$.
\end{poc}

\begin{rem}\label{rem:subfields}
The problem of finding all the subfields of some fixed degree in a given field extension $\ext{\L}{\K}$ is an active area of research. Algorithms for solving this problem can be found in \cite{EK2019,Kluners1998,KP1997,SvH2019,vHKN2011,vHKN2013}.
\end{rem}

\begin{rem}\label{rem:embedding}
Let $M$ be a quadratic extension of~$\K$ such that $\CA\otimes M$ splits. Fix an element $d \in \K$, which is a square in~$M$ but not in~$\K$. Then $M = \K(\sqrt{d})$. Now, \cite[Theoreme~I.2.8]{Vigneras1980} asserts that $M$ embeds in~$\CA$. Therefore, $d$ is a square of some pure quaternion $\qa\in \CA$. Then $a + b\sqrt{d}\mapsto a + b\qa$ is an embedding of~$M$ into~$\CA$. The quaternion~$\qa$ can be directly constructed by solving two norm equations (see, e.g., \cite{Koprowski2023}).
\end{rem}

\begin{rem}\label{rem:factorization}
All the algorithms presented in this paper rely on our ability to factor polynomials over number fields, either the base field~$\K$ or its quadratic extension~$\K_*$. Fortunately, the factorization of polynomials over number fields is a well-studied subject. We may refer the reader to \cite{Ayad2010,BHKS2009,Landau1985,Lenstra1983,MOP2005,Roblot2004}.
\end{rem}

We will now develop the main procedure of this section that takes a polynomial $\pp\in \Kx$, irreducible in every commutative subfield of~$\CA$, and constructs a quaternionic polynomial~$\qpp$ such that $\pp = \qpp\conj{\qpp}$. The general idea is to use the fact (see Proposition~\ref{prop:ireeducibility_criterium}) that the existence of a non-trivial factorization of~$\pp$ in $\Ax$ implies that $\CA\otimes \L$ splits, where as before $\L := \quo{\Kx}{\ideal\pp}$. We can then use a zero-divisor of $\CA\otimes\L$ to find a polynomial $\qpp\in \Ax$ such that $\qpp\conj{\qpp}  = \pp\pq$ for some extraneous factor~$\pq$. This can be thought of as an ``approximate factorization'' of~$\pp$. We will then successively improve this ``approximation'', reducing the degree of~$\pq$. First, we show some preliminary lemmas.

\begin{lem}\label{lem:g_divides_pr}
Let $\pp, \pq\in \Kx$ and $\qpp\in \Ax$ be polynomials and let $\qpr\in \Ax$ be the reminder of~$\qpp$ modulo~$\pq$. If $\pp\pq = \qpp\conj{\qpp}$, then $\pq$ divides the polynomials $\qpr\conj{\qpr}$ and $\qpp\conj{\qpr}$.
\end{lem}

\begin{proof} 
We have $\qpp = \qpq\cdot \pq + \qpr$ for some quaternionic polynomial $\qpq\in \Ax$. Compute 
\begin{align*}
\qpr\conj{\qpr}
&= ( \qpp - \qpq\cdot \pq)( \conj{\qpp - \qpq\cdot \pq} )\\
&= \qpp\conj{\qpp} - \qpp\conj{\qpq}\cdot \pq - \qpq\conj{\qpp}\cdot \pq + \qpq\conj{\qpq}\cdot \pq^2\\
&= \pp\pq - \qpp\conj{\qpq}\cdot \pq - \qpq\conj{\qpp}\cdot \pq + \qpq\conj{\qpq}\cdot \pq^2\\
&= (\pp - \qpp\conj{\qpq} - \qpq\conj{\qpp} + \qpq\conj\qpq\cdot \pq)\cdot \pq.
\end{align*}
This proves the first assertion. For the second one, write 
\[
\qpp\conj{\qpr} 
= (\qpq\cdot \pq + \qpr)\cdot \conj{\qpr}
= \qpq\conj{\qpr}\cdot \pq + \qpr\conj{\qpr}.
\]
Now, it follows from the previous part of the proof that the right-hand side of the above formula is divisible by~$\pq$ and so is $\qpp\conj{\qpr}$, as claimed.
\end{proof}

\begin{lem}\label{lem:f_divides_qq}
Keep the assumptions of the previous lemma and let $\qpq := \qpp\conj{\qpr}\cdot \pq^{-1}\in \Ax$. Then $\pp$ divides $\qpq\conj{\qpq}$.
\end{lem}

\begin{proof}
Compute
\begin{align*}
\qpq\conj{\qpq}
&= (\qpp\conj{\qpr}\cdot \pq^{-1})\conj{(\qpp\conj{\qpr}\cdot \pq^{-1})}\\
&= \qpp\conj{\qpr}\qpr\conj{\qpp}\cdot \pq^{-2}\\
&= \qpp\conj{\qpp}\cdot \qpr\conj{\qpr}\cdot \pq^{-2}\\
&= \pp\pq\cdot \qpr\conj{\qpr}\cdot \pq^{-2}\\
&= (\qpr\conj{\qpr}\cdot \pq^{-1}) \cdot \pp.
\end{align*}
The assertion follows now from the preceding lemma.
\end{proof}

\begin{lem}\label{lem:degree_reduction}
Keep the assumptions of Lemmas~\ref{lem:g_divides_pr} and \ref{lem:f_divides_qq}. Denote $\pr := \qpq\conj{\qpq}\cdot \pp^{-1}\in \Kx$. Then $\deg \pr < \deg \pq$.
\end{lem}

\begin{proof}
We have the following relations between the degrees of the polynomials in question:
\begin{itemize}
\item $\deg \pr = 2\cdot \deg \qpq - \deg \pp$, since $\pr = \qpq\conj{\qpq}\cdot \pp^{-1}$;
\item $\deg \qpq = \deg \qpp + \deg \qpr - \deg \pq$, since $\qpq = \qpp\conj{\qpr}\cdot \pq^{-1}$;
\item $2\cdot \deg \qpp = \deg \pp + \deg \pq$, since $\pp\pq = \qpp\conj{\qpp}$;
\item $\deg \qpr < \deg \pq$, since $\qpr$ is the reminder of~$\qpp$ modulo~$\pq$.
\end{itemize}
Combining these four conditions, we obtain the following:
\begin{align*}
\deg \pr 
&= 2\deg \qpq - \deg \pp\\ 
&= 2\deg \qpp + 2\deg \qpr - 2\deg \pq - \deg \pp\\
&= \deg \pp + \deg \pq + 2\deg \qpr - 2\deg \pq - \deg \pp\\
&= 2\deg \qpr - \deg \pq\\
&< 2\deg \pq - \deg \pq = \deg \pq.\qedhere
\end{align*}
\end{proof}

We are now ready to present an algorithm that factors in $\Ax$ an irreducible polynomial~$\pp$ with coefficients in~$\K$.

\begin{alg}\label{alg:factor_irreducible_scalar_poly}
Let $\CA = \Quat$ be a division quaternion algebra over a number field~$\K$. Given a monic irreducible polynomial $\pp\in \Kx$, this algorithm outputs its irreducible factors in $\Ax$.
\begin{enumerate}
\item Construct the field extension $\L = \quo{\Kx}{\ideal{\pp}}$.
\item\label{st:fisp:quick_exit} Check if the quaternion algebra $\CA\otimes \L$ splits. If it does not split, output $\pp$ and quit.
\item\label{st:fisp:subfield} Execute Algorithm~\ref{alg:subfield}. If it returns a factorization $(\qpp, \conj{\qpp})$ of~$\pp$, then output $(\qpp, \conj{\qpp})$ and quit.
\item\label{st:fisp:iso_vec} Find a zero-divisor $\qa = a_0 + a_1\qi + a_2\qj + a_3\qk\in \CA\otimes \L$ \textup(see Remark~\ref{rem:iso_vec}\textup).
\item\label{st:fisp:qpp} Let $\qpp\in \Ax$ be an inverse image of~$\qa$ under the natural homomorphism $\Ax\to \CA\otimes \quo{\Kx}{\ideal\pp}$, i.e. $\qpp := \pq_0 + \pq_1\cdot \qi + \pq_2\cdot \qj + \pq_3\cdot \qk$, where $\pq_0, \dotsc, \pq_3\in \Kx$ are polynomials such that $\pq_i + \ideal{\pp} = a_i$ for $i \leq 3$.
\item Set $\pq := \qpp\conj{\qpp}\cdot \pp^{-1}$.
\item While $\pq$ is non-constant do: 
  \begin{enumerate}
  \item Compute the reminder~$\qpr$ of~$\qpp$ modulo~$\pq$. 
  \item\label{st:fisp:update_p} Update $\qpp$ setting $\qpp := \qpp\conj{\qpr}\cdot \pq^{-1}$.
  \item\label{st:fisp:update_g} Update $\pq$ setting $\pq := \qpp\conj{\qpp}\cdot \pp^{-1}$.
  \end{enumerate}
\item Let $\qc := \lc \qpp$ be the leading coefficient of~$\qpp$.
\item Output $(\qpp\cdot \qc^{-1}, \conj{\qc}^{-1}\cdot \conj{\qpp})$.
\end{enumerate}
\end{alg}

\begin{poc}
Proposition~\ref{prop:ireeducibility_criterium} asserts that $\pp$ is irreducible in $\Ax$ if and only if $\CA\otimes\L$ does not split. This proves the correctness of step~\eqref{st:fisp:quick_exit}. Assume that the algorithm does not terminate at step~\eqref{st:fisp:quick_exit} or~\eqref{st:fisp:subfield}. This means that $\CA\otimes \L$ splits (i.e. it is isomorphic to $M_2\L$) and so it contains a zero-divisor $\qa = a_0 + a_1\qi + a_2\qj + a_3\qk$. Now, $\qa$ is a zero-divisor if and only if its norm vanishes (see, e.g., \cite[Exercise~3.7]{Voight2021}):
\[
\norm\qa = a_0^2 - \gena a_1^2 - \genb a_2^2 + \gena\genb a_3^2 = 0.
\]
Let the polynomial $\qpp\in \Ax$ be defined as in step~\eqref{st:fisp:qpp}. The previous formula reads now as
\[
\qpp\conj{\qpp} + \ideal{\pp}
= \pq_0^2 - \gena\pq_1^2 - \genb\pq_2^2 + \gena\genb\pq_3^2 + \ideal{\pp}
= 0.
\]
Therefore, $\pp$ divides $\qpp\conj{\qpp}$ and so $\pq$ is a well defined polynomial satisfying the condition $\pp\pq = \qpp\conj{\qpp}$. Lemma~\ref{lem:g_divides_pr} says that $\pq$ divides $\qpp\conj{\qpr}$, where $\qpr$ is the reminder of~$\qpp$ modulo~$\pq$. Hence, updating $\qpp$ in step~\eqref{st:fisp:update_p}, we again obtain a (quaternionic) polynomial. It follows from Lemma~\ref{lem:f_divides_qq} that $\pp$ divides the norm $\norm{\qpp} = \qpp\conj{\qpp}$ of the updated polynomial~$\qpp$. Consequently, after updating $\pq$ in step~\eqref{st:fisp:update_g} once again, we have a polynomial, and it clearly satisfies the condition $\pp\pq = \qpp\conj{\qpp}$. Further, Lemma~\ref{lem:degree_reduction} asserts that the degree of~$\pq$ diminishes after every iteration of the loop. It follows that the loop stops after finitely many iterations, and so the algorithm terminates. After the loop is finished we have $\pp\pq = \qpp\conj{\qpp}$ where $\pq$ is a constant. Since $\pp$ is monic by assumption, the leading coefficient of the left-hand side is $\lc( \pp\pq ) = \pq\in \K$, while the leading coefficient of the right-hand side equals $\lc(\qpp)\cdot \lc(\conj{\qpp}) = \qc\conj{\qc}$. Therefore, $\qc\conj{\qc} = \pq$ and so $\pp = \qpp\cdot \pq^{-1}\cdot \conj{\qpp} = (\qpp\cdot \qc^{-1}) (\conj{\qc}^{-1}\cdot \conj{\qpp})$ is a factorization of~$\pp$.
\end{poc}

\begin{rem}\label{rem:iso_vec}
It is well known (see \cite[Section~4]{Voight2013}) that finding a zero-divisor in a (necessarily split) quaternion algebra $\quat{\gena}{\genb}{\L}$ is equivalent to finding an $\L$-rational point on a projective conic curve given by $\gena x^2 + \genb y^2 - \gena\genb z^2 = 0$. This can be done either by a variant of a Lagrange descent (see \cite{CR2003, vHC2006} and \cite[\S127.5]{CBFS2021}) or by solving a norm equation (see \cite{Koprowski2021}). In particular, it is always possible to find a zero-divisor which is a pure quaternion, i.e. it has a form $\qa = a_1\qi + a_2\qj + a_3\qk$.
\end{rem}

\begin{exm}\label{exm:fisp}
We will illustrate how the above algorithm works using the following toy example. Fix the quaternion algebra $\CA := \quat{-1}{-1}{\QQ}$ over the rationals and let
\[
\pp = x^4 + 11x^2 + 16x + 6.
\]
The field $\L = \quo{\Kx}{\ideal{\pp}}$ splits~$\CA$, hence $\pp$ factors in~$\Ax$. The polynomial~$\qpp$ constructed in step~\eqref{st:fisp:qpp} has a form $\qpp = 0 +  \pq_1\qi + \pq_2\qj + \pq_3\qk$, where:
\begin{align*}
\pq_1 &= 19x^3 - 12x^2 + 211x + 154,\\
\pq_2 &= 13x^3 - 11x^2 + 136x + 97,\\
\pq_3 &= 53.
%1.
\end{align*}
It follows that initally the polynomial $\pq = \qpp\conj{\qpp}\cdot \pp^{-1}$ equals
\[
\pq = 530x^2 - 742x + 5989.
%\frac{10}{53} x^2 - \frac{14}{53}x + \frac{113}{53}.
\]
In this particular example, the loop is executed only once, and after it stops, we have 
\[
\qpp = 
\left(\frac{29}{50} + \frac{11}{25}\qk\right)\cdot x^2 + \left(\frac{11}{25} - \frac{13}{10}\qi + \frac{19}{10}\qj - \frac{29}{50}\qk\right)\cdot x + \frac{11}{25} - \frac{18}{25}\qi + \frac{73}{50}\qj - \frac{29}{50}\qk.
\]
Dividing $\qpp$ on the right by its leading coefficient, we obtain the sought factor of~$\pp$:
\[
\qpq_0 := \qpp\cdot \lc(\qpp)^{-1} 
= x^2 - (3\qi - \qj + \qk)\cdot x - 2\qi + \qj - \qk
\]
It can be verified by a direct computation that $\pp = \qpq_0\conj{\qpq_0}$.
\end{exm}

\section{Factoring general polynomials}
We will now present an algorithm that factors an arbitrary quaternionic polynomial. First, we prove the following auxiliary lemma that may, to some extent, mitigate inconveniences caused by the lack of commutativity.

\begin{lem}\label{lem:semicommutativity}
Let $\qpp, \qpq\in \Ax$ be irreducible monic polynomials. If the norm $\norm\qpp$ and $\norm\qpq$ are relatively prime in $\Kx$, then there are irreducible polynomials $\qpp_1, \qpq_1\in \Ax$ such that $\norm \qpp_1 = \norm \qpp$, $\norm \qpq_1 = \norm \qpq$ and $\qpp\qpq = \qpq_1\qpp_1$.
\end{lem}

\begin{proof}
The polynomials $\norm\qpp$ and $\norm\qpq$ are relatively prime by assumption. Hence, by B\'ezout identity, there are polynomials $\pp, \pq\in \Kx$ such that
\[
\pp\cdot \norm\qpp + 1 = \pq\cdot \norm\qpq.
\]
Set $\qpp_* := \conj{\qpp}\cdot \pp$ and $\qpq_* := \conj{\qpq}\cdot \pq$. Then
\[
\qpp_*\qpp + 1 = \qpq\qpq_*.
\]
This means that $\qpq$ is right-invertible\footnote{Notice that $\qpq$ is also left-invertible modulo~$\qpp$. However, to prove the latter assertion it suffices to use a weaker asumption that $\gcrd(\qpp, \qpq) = 1$ together with \cite[Theorem~I.5]{Ore1933}.} modulo~$\qpp$. Set 
\[
\qpp_1 := \qa\cdot \lclm(\qpp, \qpq_*)\cdot \qpq_*^{-1},
\]
where the quaternion~$\qa$ is selected in such a way that $\qpp_1$ is monic. Rearranging the term (see \cite[Theorem~I.13]{Ore1933}) we obtain 
\[
\qpp = \lclm(\qpp_1, \qpq)\cdot \qpq^{-1}.
\]
From the assumption $\gcd(\norm\qpp, \norm\qpq) = 1$ we infer $\gcrd(\qpp, \qpq) = 1$. Hence, the norms of~$\qpp$ and~$\qpp_1$ coincide. We have 
\[
\qpp\qpq = \lclm(\qpp_1, \qpq)\cdot \qpq^{-1}\cdot \qpq = \lclm(\qpp_1, \qpq).
\]
Therefore, the product $\qpp\qpq$ is right-divisible by~$\qpp_1$. Take $\qpq_1 := \qpp\qpq\cdot \qpp_1^{-1}$. It is clear that $\norm\qpq_1 = \norm\qpq$ and $\qpq_1\qpp_1 = \qpp\qpq$.
\end{proof}

\begin{alg}\label{alg:complete_factor}
Let it $\CA$ be a division quaternion algebra over a number field~$\K$. Given a nonzero polynomial $\qpp\in \Ax$, this algorithm outputs a quaternion $\qc\in \CA$ and a list $\CL = (\qpq_1, \dotsc, \qpq_n)$ of monic irreducible quaternionic polynomials such that $\qpp = \qc\cdot \qpq_1\dotsm \qpq_n$.
\begin{enumerate}
\item Let $\qc := \lc(\qpp)$ and update $\qpp$ setting $\qpp := \qc^{-1}\cdot \qpp$. Initialize $\CL := ()$.
\item\label{st:cf:coords} Let $\pp_0, \dotsc, \pp_4\in \Kx$ be the coordinates of~$\qpp$ with respect to the basis $\{1, \qi, \qj, \qk\}$ of~$\Ax$ treated as a free module over $\Kx$.
\item Compute the maximal central factor $\pp := \gcd( \pp_0, \dotsc, \pp_4)$ of~$\qpp$ and set $\qpq := \qpp\cdot \pp^{-1}$.
\item\label{st:cf:factor_p} Factor $\pp$ in~$\Kx$ into the product $\pp = \pr_1^{e_1}\dotsm \pr_m^{e_m}$ of irreducible polynomials $\pr_1, \dotsc, \pr_m\in \Kx$.
\item\label{st:cf:factor_qi} Factor every polynomial~$\pr_i$ in~$\Ax$ using Algorithm~\ref{alg:factor_irreducible_scalar_poly}. Append $e_i$ copies of the output of that algorithm to the list~$\CL$.
\item Factor the norm $\norm\qpq = \qpq\conj{\qpq}$ in $\Kx$ into the product $\norm\qpq = \pq_1^{\varepsilon_1}\dotsm \pq_k^{\varepsilon_k}$ of monic irreducible polynomials $\pq_1, \dotsc, \pq_k\in \Kx$.
\item\label{st:cf:ext_loop} Repeat the following steps as long as the polynomial~$\qpq$ remains non-constant: 
  \begin{enumerate}
  \item\label{st:cf:gcrd} Compute the greatest common right divisor $\qpr := \gcrd(\qpq, \pq_k)$. 
  \item Update the list~$\CL$ prepending $\qpr$ to it at the beginning.
  \item Update the polynomial~$\qpq$ setting $\qpq := \qpq\cdot \qpr^{-1}$.
  \item Decrement the exponent~$\varepsilon_k$ by~$1$. 
  \item If $\varepsilon_k = 0$, then decrement the index~$k$ by~$1$.
  \end{enumerate}
\item Output $\qc$ and the list~$\CL$.
\end{enumerate}
\end{alg}

\begin{poc}
Fix a nonzero polynomial $\qpp\in \Ax$. Proposition~\ref{prop:Beck} asserts that $\qpp$ can be uniquely expressed as $\qpp = \qc\cdot \qpq\cdot \pp$, where $\pp$ is the greatest common divisor (in $\Kx$) of the coordinates of~$\qpp$, constructed in step~\eqref{st:cf:coords} of the algorithm. It is clear that by factoring~$\qpq$ and~$\pp$ into products of irreducible quaternionic polynomials, we obtain a factorization of the original polynomial~$\qpp$. Factorization of~$\pp$ is performed in steps (\ref{st:cf:factor_p}--\ref{st:cf:factor_qi}) using Algorithm~\ref{alg:factor_irreducible_scalar_poly}. It remains to factor~$\qpq$. Assume that $\qpq = \qpr_1\dotsm \qpr_l$ is some factorization of~$\qpq$. Then $\norm\qpq = \norm{\qpr_1}\dotsm \norm{\qpr_l}$ is a factorization in~$\Kx$ of the norm of~$\qpq$. But now, $\Kx$ is a (commutative) unique factorization domain, hence every $\norm{\qpr_i}$ must equal precisely one of the irreducible factors $\pq_1, \dotsc, \pq_k\in \Kx$ of~$\norm\qpq$. Repeatedly applying Lemma~\ref{lem:semicommutativity} we see that there is another factorization $\qpq = \qpq_1\dotsm \qpq_l$ such that $\norm{\qpq_l} = \pq_k$. Consequently, we can extract the right-most factor $\qpq_l = \gcrd(\qpq, \pq_k)$. A simple induction shows that the loop in step~\eqref{st:cf:ext_loop} will eventually produce all the irreducible factors.
\end{poc}

\begin{rem}
One may wonder why we have to use a dedicated method when dealing with polynomials having coefficients in the center of the algebra. The reason is that the norm of a central polynomial $\pp \in \Kx$ is just its square $\norm{\pp} = \pp^2$. Consequently $\gcrd(\pp, \norm\pp) = \gcd(\pp, \pp^2) = \pp$. Hence, even if $\pp$ factors in $\Ax$ into a product $\pp = \qpp\conj{\qpp}$, step~\eqref{st:cf:gcrd} of Algorithm~\ref{alg:complete_factor} cannot extract its factors. The same phenomenon is observed when factoring polynomials in a finite field extension using Trager's algorithm. Say, $\L = \K(\vartheta)$ is a finite extension of~$\K$ and $\pp\in \L[x]$. In Trager's algorithm, to factor~$\pp$, we ``shift'' it and compute the norm of $\pp(x + k\vartheta)$ for some $k\in \ZZ$. Unfortunately, in the case of quaternionic polynomials, this trick no longer works due to the lack of commutativity. That is why we need to  recourse to the more complex method described in Section~\ref{sec:central}. Therefore, factorization of polynomials over division quaternion algebras (hence algebras of degree two) is much more time-consuming than factorization over degree two field extensions. Nonetheless, detailed complexity analysis needs to be a subject of further research.
\end{rem}

\begin{exm}\label{exm:complete_factor}
We will illustrate Algorithm~\ref{alg:complete_factor} with another toy example. Take again the quaternion algebra $\CA = \quat{-1}{-1}{\QQ}$ and let
\begin{multline*}
\qpp :=
(1 + \qk)\cdot x^8 + (-2 + \qi - \qj - 2\qk)\cdot x^7 + (9 + 11\qk)\cdot x^6\\ 
+ (-6 + 12\qi - 6\qj - 2\qk)\cdot x^5 + (-45 + 10\qi - 18\qj - 27\qk)\cdot x^4\\
+ (-44 + 17\qi + 49\qj + 32\qk)\cdot x^3 + (21 - 50\qi + 58\qj + 53\qk)\cdot x^2\\ 
+ (48 - 90\qi - 2\qj + 8\qk)\cdot x + 18 - 36\qi - 12\qj - 6\qk.
\end{multline*}
Then $\qpp = (1 + \qk)\cdot (\pp_0 + \pp_1\qi + \pp_2\qj + \pp_3\qk)$, where
\begin{align*}
\pp_0 &= x^8 - 2x^7 + 10x^6 - 4x^5 - 36x^4 - 6x^3 + 37x^2 + 28x + 6,\\
\pp_1 &= 3x^5 - 4x^4 + 33x^3 + 4x^2 - 46x - 24,\\
\pp_2 &= -x^7 - 9x^5 - 14x^4 + 16x^3 + 54x^2 + 44x + 12,\\
\pp_3 &= x^6 + 2x^5 + 9x^4 + 38x^3 + 16x^2 - 20x - 12.
\end{align*}
Computing the greatest common divisor of these four polynomials, we obtain the maximal central factor of~$\qpp$:
\[
p := \gcd(p_0, \dotsc, p_4) = x^4 + 11x^2 + 16x + 6.
\]
The polynomial~$\pp$ is precisely the one that was factored already in the previous example. It remains to factor the polynomial 
% \begin{align*}
% \tilde\qpp &:= (1+\qk)^{-1}\cdot \qpp\cdot p^{-1}\\ 
% &= x^4 + (-2 - \qj)\cdot x^3 + (-1 + \qk)\cdot x^2 + (2 + 3\qi + 2\qj + 2\qk)\cdot x + 1 - 4\qi + 2\qj - 2\qk, 
% \end{align*}
$\qpq := (1+\qk)^{-1}\cdot \qpp\cdot p^{-1},$
that is not divisible by any polynomial with coefficients in~$\K$. The norm of $\qpq$ factors in~$\Kx$ into a product $\norm\qpq = \pq_1\pq_2\pq_3$, where
\begin{align*}
\pq_1 &= x^2 + 1,\\
\pq_2 &= x^2 - 4x + 5,\\
\pq_3 &= x^4 - 3x^2 + 5.
\end{align*}
Computing the successive greatest common right divisors, we obtain factors of~$\qpq$:
\begin{align*}
\qpq_3 &:= \gcrd(\qpq, \pq_3) = x^2 + \qi x -2 - \qk,\\
\qpq_2 &:= \gcrd(\qpq\cdot \qpq_3^{-1}, \pq_2) = x - 2 - \qj,\\
\qpq_1 &:= \gcrd(\qpq\cdot \qpq_3^{-1}\cdot \qpq_2^{-1}, \pq_1) = x - \qi.
\end{align*}
This yields a factorization $\qpp = (1 + \qk)\cdot \qpq_1\qpq_2\qpq_3\cdot \qpq_0\conj{\qpq_0}$, where $\qpq_0$ is the factor of~$\pp$ constructed in Example~\ref{exm:fisp}.
\end{exm}

\section{Root finding}
Polynomial factorization is closely related to the problem of root finding. For quaternionic polynomials, the latter question was studied already 80 years ago by Niven in the case of Hamiltonian quaternions (see \cite{Niven1941}) and 40 years ago by Beck for arbitrary quaternion algebras (see \cite{Beck1979}). Nevertheless, so far, computational methods have concentrated exclusively on Hamiltonian quaternions (see, e.g., \cite{JO2010, Kalantari2013, Niven1941, SKS2019, SP2001}). In this section, we show how to adapt Algorithm~\ref{alg:complete_factor} to find roots in~$\CA$ of a polynomial $\qpp\in \Ax$, where $\CA$ is a division quaternion algebra over a number field.

It is known that the Fundamental Theorem of Algebra holds for the Hamiltonian quaternions (see, e.g., \cite{Kalantari2013,Niven1941,SKS2019}). Hence, every non-constant polynomial with coefficients in $\HH = \quat{-1}{-1}{\RR}$ has a root in~$\HH$. This property no longer holds if we replace $\HH$ by a split quaternion algebra over the reals. Similarly, it is completely unsurprising that, in the situation considered in this paper, there are polynomials with coefficients in division quaternion algebras over number fields that do not have roots in their rings of coefficients. If it is the case, then the algorithm must of course return an empty set.

In a commutative world, finding roots in $\L\supseteq \K$ of a polynomial $\pp\in \Kx$ is equivalent to finding linear factors of~$\pp$ in $\L[x]$. It no longer works when we leave the commutative playpen and consider polynomials over non-commutative division rings. Here we encounter two immediate obstacles. First, a quaternionic polynomial may have infinitely many zeros. Secondly, the evaluation map $\Ax\to \CA$, $\qpp\mapsto \qpp(\qa)$ for some fixed~$\qa$ is no longer a ring homomorphism if $\qa\notin Z(\CA)$. In particular, $\qpp = \qpq\qpr$ and $\qpq(\qa) = 0$ does not yet imply that $\qpp(\qa) = 0$. Fortunately, both these inconveniences are easy to overcome. First, however, for the reader's convenience, let us cite two very classical facts concerning the roots of quaternionic polynomials. Recall that two nonzero quaternions $\qa$, $\qb$ are said to be \term{conjugate} if there is a quaternion $\qc\in \CA^\times$ such that $\qa = \qc\qb\qc^{-1}$. Clearly, conjugacy is an equivalence relation.

\begin{thm}[Gordon--Motzkin]\label{thm:GM}
Let $\CA$ be a division quaternion algebra and $\qpp\in \Ax$. Then:
\begin{enumerate}
\item A quaternion $\qa\in \CA$ is a root of~$\qpp$ if and only if $\qpp$ is right-divisible by $x - \qa$.
\item At most $\deg \qpp$ conjugacy classes of~$\CA$ contains roots of~$\qpp$.
\item If $\qpp = (x - \qa_1)\dotsm (x - \qa_n)$, then every root of~$\qpp$ is conjugate to some~$\qa_i$. 
\item If a conjugacy class of~$\CA$ contains more than one root of~$\qpp$, then $\qpp$ has infinitely many roots in this class. 
\end{enumerate}
\end{thm}

For a proof of the above theorem we refer the reader to \cite{GM1966} or \cite[\S16]{Lam2001}. For a nonzero quaternion $\qa\in \CA$ we denote its \term{characteristic polynomial} by $\charpoly\qa = \norm{(x - \qa)} = x^2 - \Tr(\qa)\cdot  x + \norm{(\qa)}\in \Kx$. It is obvious that $\charpoly\qa = \charpoly{\conj{\qa}} = \charpoly\qb$ for every conjugate~$\qb$ of~$\qa$. Conversely, Dickson's theorem (see e.g., \cite[Theorem~16.8]{Lam2001}) asserts that $\charpoly\qa = \charpoly\qb$ implies that $\qa$ and $\qb$ are conjugate. The next theorem is taken from~\cite{Beck1979}.

\begin{thm}[Beck]\label{thm:Beck}
Let $\CA$ be a division quaternion algebra and $\qpp\in \Ax$ a non-constant polynomial. Then:
\begin{enumerate}
\item A quaternion $\qa$ is conjugate to a root~$\qb$ of~$\qpp$ if and only if $\charpoly{\qa}$ divides $\norm{\qpp}$ in $\Kx$ if and only if $\qa$ is a root of $\norm{\qpp}$.
\item If $\charpoly\qa$ divides $\qpp$, for some $\qa\in \CA^\times$ then every conjugate of~$\qa$ is a root of~$\qpp$.
\end{enumerate}
\end{thm}

Fix a nonzero quaternionic polynomial $\qpp\in\Ax$. Write it again as a product $\qpp = \qc\cdot \qpq\cdot \pp$, where $\qc$ is the leading coefficient, $\pp$ the maximal central factor of~$\qpp$, and $\qpq$ is not divisible by any non-constant polynomial from $\Kx$. In view of the above two theorems, it is clear that to find roots of~$\qpp$, all we need to do is to find its linear right factors. To this end, it suffices to consider just the linear and quadratic factors of~$\pp$ together with the quadratic factors of $\norm\qpq$. Moreover, since $\qpp$ may have infinitely many roots, the best we can do is to enumerate roots up to conjugacy relation. This is precisely what Algorithm~\ref{alg:root_finding} does. First, however, we need to state the following lemma, which is widely known, but we are not aware of any convenient reference.

\begin{lem}\label{lem:charpoly_divides_p}
A quaternion $\qa\in \CA\setminus \K$ is a root of a polynomial $\pp\in \Kx$ if and only if $\charpoly{\qa}$ divides~$\pp$.
\end{lem}

\begin{proof} 
Assume that $\pp(\qa) = 0$. Then \cite[Theorem~16.6]{Lam2001} asserts that $\pp = \qpq\cdot \charpoly\qa$ for some polynomial $\qpq\in \Ax$. But $\pp$ and $\charpoly\qa$ sit in $\Kx$, hence also $\qpq\in \Kx$ (see e.g., \cite[Exercise~16.1]{Lam2001}). The opposite implication is trivial. 
\end{proof}

We are now ready to present the root finding algorithm, which is an easy adaptation of the factorization procedure presented in Algorithm~\ref{alg:complete_factor}.

\begin{alg}\label{alg:root_finding} 
Given a nonzero polynomial $\qpp\in \Ax$, this algorithm outputs a set $\CR\subset \CA$ such that every element of~$\CR$ is a root of~$\qpp$ and every root of~$\qpp$ is conjugate to a unique element of~$\CR$. 
\begin{enumerate}
\item Let $\qc := \lc(\qpp)$ and update $\qpp$ setting $\qpp := \qc^{-1}\cdot \qpp$.
\item Let $\pp_0, \dotsc, \pp_4\in \Kx$ be the coordinates of~$\qpp$ with respect to the basis $\{1, \qi, \qj, \qk\}$ of~$\Ax$ treated as a free module over $\Kx$.
\item Compute the maximal central factor $\pp := \gcd( \pp_0, \dotsc, \pp_4)$ of~$\qpp$ and set $\qpq := \qpp\cdot \pp^{-1}$.
\item Factor $\pp$ in~$\Kx$ into the product $\pp = \pr_1^{e_1}\dotsm \pr_m^{e_m}$ of irreducible polynomials $\pr_1, \dotsc, \pr_m\in \Kx$.
\item\label{st:rf:roots_in_K} Initialize $\CR$ to be the set of roots \textup(in~$\K$\textup) of all the linear factors of~$\pp$:
\[
\CR := \bigl\{ a\in \K : \pr_i = x - a\text{ for some }i\leq k\bigr\}.
\]
\item\label{st:rf:factors_of_p} For every polynomial $\pr_i$ of degree~$2$:
  \begin{enumerate}
  \item Execute Algorithm~\ref{alg:subfield} with input~$\pr_i$, if it reports a failure, then reiterate the loop.
  \item\label{st:rf:roots_of_p} Otherwise, if Algorithm~\ref{alg:subfield} returns a factorization $\pr_i = (x - \qa_i)(x - \conj{\qa_i})$, then check if $\CR$ contains an element conjugate to~$\qa_i$.
  \item If it does not, then add $\qa_i$ to~$\CR$.
  \end{enumerate}
\item Factor the norm $\norm\qpq = \qpq\conj{\qpq}$ in $\Kx$ into the product $\norm\qpq = \pq_1^{\varepsilon_1}\dotsm \pq_k^{\varepsilon_k}$ of monic irreducible polynomials $\pq_1, \dotsc, \pq_k\in \Kx$. 
\item For every polynomial~$\pq_j$ of degree~$2$: 
  \begin{enumerate}
  \item\label{st:rf:roots_of_q} Compute the greatest common right divisor $x - \qb_j = \gcrd(\pq_j, \qpq)$.
  \item Check if $\CR$ contains an element conjugate to~$\qb_j$.
  \item\label{st:rf:add_root_of_q} If it does not, then add~$\qb_j$ to~$\CR$.
  \end{enumerate}
\item Output the set~$\CR$.
\end{enumerate}
\end{alg}

\begin{poc}
Fix a quaternionic polynomial $\qpp = \qc\cdot \qpq\cdot \pp$ and let~$\CR$ be the set constructed by the algorithm. First, we will prove that every element of~$\CR$ is a root of~$\pp$. If an element $a\in \CR$ was constructed in step~\eqref{st:rf:roots_in_K}, then it is clear that $\qpp(a) = \pp(a) = 0$. Likewise, if $\qa_i\in \CR$ was constructed in step~\eqref{st:rf:roots_of_p}, then $\charpoly{\qa_i} = (x - \qa_i)(x - \conj{\qa_i})$ divides~$\pp$. Hence, $\qa_i$ is again a root of~$\qpp$ (and so are all its conjugates by Theorems~\ref{thm:Beck}). Now, take an element $\qb_j\in \CR$ constructed in step~\eqref{st:rf:roots_of_q}. Then $x - \qb_j$, being the greatest common right divisor of~$\qpq$ and~$\pq_j$, is, in particular, a right divisor of~$\qpp$, since $\pp$ is central. It follows that $\qb_j$ is a root of~$\qpp$. 

Conversely, we will now show that every root of~$\qpp$ is conjugate to some element of~$\CR$. Assume that $\qa\in \CA$ is a root of~$\qpp$. Then $x - \qa$ is a right divisor of~$\qpp$, by Theorem~\ref{thm:GM}. We need to consider three cases. If $\qa$ sits in~$\K$, then $x - \qa$ divides~$\pp$ and so $\qa$ is one of the elements constructed in step~\eqref{st:rf:roots_in_K} of the algorithm. Hence, without loss of generality, we can assume that $\qa\notin\K$. If $x - \qa$ divides~$\pp$, then the characteristic polynomial~$\charpoly{\qa}$ of~$\qa$ divides~$\pp$, by Lemma~\ref{lem:charpoly_divides_p}. Thus, $\charpoly{\qa}$ is one of the factors $\pr_i$ of degree~$2$ considered in step~\eqref{st:rf:factors_of_p}. Consequently, either $\qa$ is added to~$\CR$, or $\conj{\qa}$ is added to~$\CR$, or $\CR$ already contains an element conjugate to one of these two. But, $\qa$ and $\conj{\qa}$ are conjugate by Dickson's theorem (see e.g., \cite[Theorem~16.8]{Lam2001}). Therefore, either way, $\CR$ contains an element conjugate to~$\qa$. 

Finally, assume that $x - \qa$ does not divide~$\pp$. It follows that $x - \qa$ must be a right divisor of~$\qpq$ since $\pp$ is central. But then the norm $\norm(x - \qa) = \charpoly{\qa}$ is irreducible in $\Kx$ and divides $\norm{\qpq}$. Consequently, $\charpoly{\qa} = \pq_j$ for some $j\leq k$ and so $x - \qa = \gcrd( \qpq, \pq_j)$. Therefore, either $\qa$ itself is added to~$\CR$ in step~\eqref{st:rf:add_root_of_q} or $\CR$ contains an element conjugate to~$\qa$. This way, we have proved that every root of~$\qpp$ is conjugate to some element of~$\CR$. This element is uniquely determined since the elements of~$\CR$ are pairwise non-conjugate. 
\end{poc}

\section{Conclusion and further studies}
We have presented an explicit algorithm for factorization of unilateral polynomials with coefficients in division quaternion algebras over number fields. As mentioned earlier, in general, a quaternionic polynomial may have infinitely many different factorizations into irreducible factors. The algorithm presented in the paper outputs just one of these factorizations. This fact makes the complexity of the algorithm rather tricky. It has been empirically observed by the author that the output can be ``suboptimal'' in the sense that the sizes of coefficients can become arbitrarily large. This phenomenon occurs mostly when factoring central polynomials since Algorithm~\ref{alg:factor_irreducible_scalar_poly} may suffer from the explosive growth of coefficients that is typical to Euclidean-like algorithms over polynomial rings. Methods to overcome this obstacle by selecting factors with small coefficients (or restricting the growth of the coefficients) should be a subject for further studies.

The algorithms introduced in the paper were implemented by the author in a package  \qPoly{} for the computer algebra system \Magma{}. The package is distributed under the MIT license and can be freely downloaded from the author's web page \url{http://www.pkoprowski.eu/qpoly}. 

%\bibliographystyle{plain} 
%\bibliography{qpoly_factor}

\end{document}